\documentclass[amsmath,amssymb,aps,prd,longbibliography,nofootinbib,12pt]{revtex4-2}

\usepackage{graphicx}
\usepackage{color}
\usepackage[utf8]{inputenc}
\usepackage[colorlinks,hyperindex]{hyperref}  
\hypersetup{
   colorlinks,%
   citecolor=blue,%
   linkcolor=blue,%
   urlcolor=blue,%
 }

\usepackage{accents}
\usepackage{enumitem}
\usepackage{xcolor}

\def\setR{\mathbb{R}}

\newcommand{\sss}[1]{\scriptscriptstyle #1}

\usepackage{amsthm}
    \theoremstyle{plain}
        \newtheorem{theorem}{Theorem}
        \newtheorem{proposition}{Proposition}

    \theoremstyle{definition}

    \theoremstyle{remark}

\begin{document}

\title{SO$(2,n)$-compatible embeddings of conformally flat $n$-dimensional submanifolds in $\setR^{n+2}$}

\author{E.~Huguet}
\affiliation{ Universit\'e Paris Cité, APC-Astroparticule et Cosmologie (UMR-CNRS 7164), 
Batiment Condorcet, 10 rue Alice Domon et L\'eonie Duquet, F-75205 Paris Cedex 13, France.}
\email{huguet@apc.univ-paris7.fr}
\author{J.~Queva}
\affiliation{Universit\'e de Corse -- CNRS UMR 6134 SPE, Campus Grimaldi BP 52, 20250 Corte, France.}
\email{queva@univ-corse.fr}
\author{J.~Renaud}
\affiliation{Universit\'e Gustave Eiffel, APC-Astroparticule et Cosmologie (UMR-CNRS 7164), 
Batiment Condorcet, 10 rue Alice Domon et L\'eonie Duquet, F-75205 Paris Cedex 13, France.} 
\email{jacques.renaud@u-pem.fr}

\begin{abstract}
    
    We describe embeddings of $n$-dimensional Lorentzian  manifolds, including Friedmann-Lemaître-Robertson-Walker spaces, in $\setR^{n+2}$
    such that the metrics of the submanifolds are inherited by a restriction from that of $\setR^{n+2}$, and the action of the linear group SO$(2,n)$ of the ambient space reduces to conformal transformations on the submanifolds.

\end{abstract}

\maketitle

\section{Introduction}
 It is well known that de Sitter space of dimension $n$ can be obtained
as a hyperboloid embedded in a $\setR^{n+1}$ Minkowski space. This short note
presents an observation to realize more general cosmological
Friedmann-Lema\^itre-Roberston-Walker (FLRW) spacetimes in a similar fashion in $\setR^{n+2}$, compatible with the action of the conformal group SO$(2,n)$.

Embeddings of conformally flat $n$-dimensional Lorentzian manifolds in higher-dimensional flat spacetimes have been considered in the past in various contexts \cite{Pavsic:2000qy,Seahra:2002ib,Gulamov:2011ux,
Paston:2013uia,PoncedeLeon:2015pug,Akbar:2017vja,Dunajski:2018xoa}. 
More often than not, these embeddings use convenient coordinate systems, in which metric the tensor of the embedded manifold can be extracted from the ambient metric by applying suitably chosen constraints.
The advantage of such a method is that it is fitted for specific practical applications. By contrast, some aspects, such as the action of the linear conformal group (when applicable) or a global geometric view, are much less apparent.

In this paper, in a coordinate free approach, we build the embedding in $\setR^{n+2}$ of $n$-dimensional conformally flat spaces, hereafter denoted by $W$, including FLRW spacetimes as special cases.  
While explicit embedding formulas date back to FLRW space's infancy \cite{Robertson:1928,Robertson:1933zz} and were rediscovered unknowingly later \cite{Rosen:1965},
here the embedding is as natural as possible, meaning that the metric on $W$ is the restriction of the  $\setR^{n+2}$ metric, and the action of the linear group SO$(2,n)$ of $\setR^{n+2}$ reduces to conformal transformations on $W$.
The present work originates from the need of a geometric  (coordinate free) framework for the generalization of  previous works regarding the restriction to submanifolds of differential operators in particular, the Laplace operator \cite{Huguet:2022rxi}.

The geometric setting with its definitions and conventions is exposed in Sec.~\ref{Sec-GeometricSettings}. The embeddings follow from Proposition~\ref{Prop-Main}, proved in Sec.~\ref{Sec-Embeddings}. 
The inverse problem, extensions, and relations to previous works are discussed in Sec.~\ref{Sec-Remarks}.

\section{Geometric setting}\label{Sec-GeometricSettings}

Let $\setR^{n+2}$ be a pseudo-Euclidean space of dimension $n+2$ equipped with the metric $\eta =+-\cdots -+$.
Throughout the paper,
$\mu, \nu, \ldots = 0, \ldots, n-1$  are related to $n$-dimensional manifolds,
and  $\alpha, \beta, \ldots = 0, \ldots, n+1$ 
to $\setR^{n+2}$.
The canonical coordinates of a point $y$ of $\setR^{n+2}$  
are denoted $\{y^{ \alpha}\}$,
and the associated Cartesian orthonormal basis is denoted $\{\partial_\alpha\}$. We denote by $D$ the dilation operator ; in  $y$ coordinates it reads $D = y^\alpha \partial_\alpha$.

We will denote by $X_h$ the $n$-dimensional  manifold obtained as the intersection 
of the $(n+1)$-dimensional null  cone of $\setR^{n+2}$:  $\mathcal{C}:=\{y\in \setR^{n+2}\, ;\,  C(y):=y^\alpha y_\alpha =  0\}$, 
and the surface  $P_h:= \{y\in \setR^{n+2} \, ;\,  h(y) = 1\}$, $h$ being a homogeneous function of degree one. 

\section{Embeddings}\label{Sec-Embeddings}

The SO$(2,n)$-compatible embedding of ($n$-dimensional) FLRW spaces  in $\setR^{n+2}$ is a special case of the following theorem :
\begin{theorem}\label{TH-theorem}
    Let $\setR^{n+2}$ be the $n+2$ dimensional real space with the pseudo-Euclidean metric $\eta$ with signature $(2,n)$. 
    Let $f$ and $l$ be two homogeneous functions of degrees one and zero, respectively, and $k=e^{-l}f$. 
    Let $g^f$ be the induced metric from $\setR^{n+2}$ to $X_f$, 
    with $g^k$ having the same relation with respect to $X_k$.
    Then, $g^k=e^{2l}{\tilde g}^f$, where ${\tilde g}^f$ is a metric on $X_k$ which is isometric to $g^f$.
    Moreover, the elements of the linear group SO$(2,n)$ on $\setR^{n+2}$ act on $X_k$ and $X_f$ as conformal transformations.
\end{theorem}
\begin{proof}
 Let  
\begin{align*}
\Lambda: \setR^{n+2} &\to\setR^{n+2} \\
y &\mapsto e^{l(y)}y.
\end{align*}
We clearly have  $\Lambda X_f=X_k$, 
since, for $y\in X_f$, one has $e^{l(y)}y=e^{l(y)}y/f(y)=y/k(y)$ and $k\left(y/k(y)\right) = 1$. Moreover, $\Lambda$ induces a diffeomorphism $\Lambda_r$ between $X_f$ and $X_k$.

In order to determine the relation between $g^{f}$ and $ g^{k}$, the metrics induced from the ambient space on $X_f$ and on $X_k$, respectively,
let us consider $V \in T_y\setR^{n+2}$.
A straightforward calculation from the definition of the push forward leads to 
\begin{equation*}
\Lambda_*V =e^l(  V + \langle dl,V\rangle D).
\end{equation*}
This map induces the map $\Lambda_{r*}$ between $TX_f$ and $TX_k$, with the same expression on a vector field of  $TX_f$.

Now,  let $U'=\Lambda_{r*}U$,  $V'=\Lambda_{r*} V$, where $U,V\in TX_f$, and let $m$ and $n$ be the canonical injections 
from  $X_f$ and $X_k$, respectively, in $\setR^{n+2}$.
Then, noting that $\eta(D,U)=\eta(D,V)=\eta(D,D)=0$, one has successively
\begin{align*}
     g^k_{ \Lambda(y)}(U',V')&=(n^*\eta)_{ \Lambda(y)}(U',V')\\
    &=\eta(\Lambda_{r*}U,\Lambda_{r*}V)\\
    &=\eta\large(e^l(  V + \langle dl,V\rangle D), e^l(  V + \langle dl,V\rangle D)\large)\\
    &=e^{2l}\eta(U,V)\\
    &=e^{2l} \eta(m_*U,m_*V)\\
    &= e^{2l} m^*\eta(U,V)\\
    &=e^{2l}g^f(U,V).
\end{align*}
The map $\Lambda_{r*}$ is thus an isometry between $(X_f,e^{2l}g^f)$ and $(X_k,g^{k})$, proving the first assertion of the theorem.

We now consider the SO$(2,n)$ action on $X_k$. 
Let the action of SO$(2,n)$ on the set $X_k \subset \setR^{n+2}$ be
\begin{equation*}
    \alpha^k(y)=\frac{\alpha. y}{k(\alpha. y)}\ \in X_k,\quad
    \alpha\in \mathrm{SO}(2,n),
\end{equation*}
where $\alpha . y$ is the natural SO$(2,n)$ action in $\setR^{n+2}$.
We claim that this action on $X_k$  is a conformal transformation.

In fact, the action $(\alpha^k)'$ on $X_k$ is the tangent one. For $V \in TX_k$ one has
\begin{equation*}
    (\alpha^k)'(y)(V)=\frac{1}{k(\alpha.y)}\alpha. V
    -\frac{k'(\alpha. y)(\alpha.V)}{k^2(\alpha. y)}(\alpha. y).
\end{equation*}
Since $\alpha$ is isometric with respect to $\eta$, and $y$, viewed as a vector of $\setR^{n+2}$, is perpendicular to $V\in TX_k$ ---that is, $\eta(y,V)=\eta(y,y)=0$--- one has
\begin{equation*}
    \eta((\alpha^k)'(y)(V_1),(\alpha^k)'(y)(V_2))=
\frac{1}{k^2(\alpha. y)}\eta (V_1,V_2),\quad
    V_1, V_2 \in TX_k.
\end{equation*}
That is, $\alpha^k$ acts as a conformal transformation on $X_k$.
This completes the proof of the theorem.
\end{proof}

We now return to the problem of the natural embedding of FLRW spaces in  $\setR^{n+2}$.
We begin with the $n$-dimensional de~Sitter (dS) space $\Sigma$, with metric $g^{dS}$ and scalar (Ricci) curvature $R = - n(n-1) H^2$.
Let us define the FLRW space $W$ through the scale factor $a$ such that $g^W=e^{2a}g^{dS}$.

The de~Sitter space $\Sigma$ is first realized, as usual~\cite{Hawking:1973uf},  as the hypersphere of the pseudo-Euclidean space $\setR^{n+1}$ through the equation
\begin{equation*}
    y^\mu y_\mu- (y^n)^2= -H^{-2}. 
\end{equation*}
We then identify $\setR^{n+1}$ with the (hyper)plane of $\setR^{n+2}$ defined through $f(y)=Hy^{n+1}=1$. This realizes an isometric embedding of $\Sigma$ in $\setR^{n+2}$ (see 
\cite{Huguet:2022rxi} for a more general approach).
The key point is that $\Sigma$ is now defined as the intersection of the null cone of $\setR^{n+2}$ and the plane $f=1$:
\begin{equation*}
\left\{
\begin{aligned}
    C(y) &= y^\mu y_\mu -  (y^n)^2 +   (y^{n+1})^2 = 0,\\
    f(y) &= H y^{n+1} = 1.
\end{aligned}
    \right.
\end{equation*}
In other words, $\Sigma=X_f$.

\begin{proposition}\label{Prop-Main}
    Let $l$ be a homogeneous function of degree zero on $\setR^{n+2}$ whose restriction to $\Sigma$ is the function $a$.
    Let $W=e^l\Sigma$, then the induced metric on $W$ by $\eta$ of $\setR^{n+2}$ is the Weyl rescaling $g^{\sss W}=e^{2a}g^{\sss dS}$ of the de~Sitter metric $g^{\sss dS}$.  Moreover, the elements of the linear group SO$(2,n)$ on $\setR^{n+2}$ act on $W$ as conformal transformations.
\end{proposition}
\begin{proof}
    This proposition is nothing but a specialization of the above theorem with $\Sigma=X_f$ and $W=X_k$.  
\end{proof}

One can ask why we could not specialize $l$ so that it does not depend on $y^{n+1}$,
eliminate that variable, and then obtain an embedding in $\setR^{n+1}$.
In that case,  the metric on $W$ would no longer be induced by that of $\setR^{n+1}$, and SO$(2,n)$ invariance
would be lost. Indeed, defining $W$ as a subset of the $\setR^{n+2}$ cone is key to the proof through isotropic vectors. 

Note that the same result can be obtained by continuous deformation of Minkowski or anti-de~Sitter (AdS) spaces instead of de~Sitter space. 
This amounts to choosing $f=Hy^n$ or $f=H(y^n+y^{n+1})/2$, respectively.

\section{Concluding remarks}\label{Sec-Remarks}

In the previous section, Sec.~\ref{Sec-Embeddings}, we built $k$, and thus the manifold $X_k$, from $f$ and the scale factor $a$.
One can ask about the converse problem that is, starting from a function $k$, homogeneous of degree one, is the manifold $X_k$ the continuous deformation 
of some de~Sitter space $X_f$, where $f$  is a homogeneous polynomial of degree one?
Globally, the answer is most likely negative.
Locally, however, $X_k$ can be obtained as a continuous deformation of dS or AdS spaces whose related defining planes $P_f$ are tangent to $X_k$. Precisely, let us 
consider a point $y_o$ of $X_k$ and set $K_o :=\nabla k(y_o)$: we assume that $(K_o)^2 \neq 0$, and the tangent plane 
at $y_o$ is defined through  $f(y)= K_o\cdot (y-y_o)+k(y_o)$. Then, 
from Theorem \ref{TH-theorem}, the metric $g^k$ is, locally, a continuous deformation of a dS or AdS metric depending on the sign of $(K_o)^2$ and hence
on the point $y_o$ considered on the manifold.

Concerning the scale factor $a$, it can be chosen as any continuous function of $y \in \Sigma$. As a consequence, by contrast with 
the scale factor appearing in the metric of FLRW spacetimes in (four-dimensional) cosmology (see, for instance, \cite{Ibison:2007dv} about different forms of the RW metric), it can describe
more general conformally flat spacetimes, in particular those with less symmetry than the FLRW space. 

In the most general, case the need to embed FLRW spaces in $\setR^{n+2}$, in contrast to simply $\setR^{n+1}$, has been pointed out in \cite{Akbar:2017vja}.
Here, in our works, considering $\setR^{n+2}$ to boot stems from the need to track how the SO$(2,n)$ group acts.


\end{document}